\pdfoutput=1

\documentclass[12pt, a4paper]{article} 

\usepackage[german, english]{babel}
\usepackage[latin1]{inputenc}
\usepackage[T1]{fontenc}
\usepackage{lmodern}
\usepackage{amsmath}    
\usepackage{amsthm}  
\usepackage{amsfonts}     
\usepackage{amssymb}
\usepackage{array}
\usepackage[affil-it]{authblk}
\usepackage{caption}
\usepackage{color}
\usepackage{courier}
\usepackage{diagbox}
\usepackage{dsfont}			
\usepackage{enumerate}			 
\usepackage{epsfig}
\usepackage{float}      
\usepackage[top=2cm,left=2cm,right=2cm,bottom=2cm]{geometry} 
\usepackage{graphicx}   
\usepackage{cite}
\usepackage{listings}        
\usepackage{lastpage} 
\usepackage{mdframed}
\usepackage{multicol}  
\usepackage{newclude}
\usepackage{placeins}
\usepackage{scrpage2}
\usepackage{setspace}
\usepackage{shuffle}
\usepackage{subcaption}
\usepackage[tc]{titlepic}
\usepackage{tikz}
\usepackage{tikz-cd}
\usetikzlibrary{arrows,arrows.meta,shapes,trees,decorations.pathmorphing,decorations.markings,backgrounds,calc,positioning}
\usepackage{titlesec}
\usepackage{url}
\usepackage{varwidth}
\usepackage{verbatim}
\usepackage{wrapfig}
\usepackage{xcolor}
\usepackage{CJK}
\usepackage{hyperref}

\newcommand*{\fermion}{ \includegraphics[height=0.5\baselineskip]{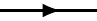} }
\newcommand*{\photon}{ \includegraphics[height=0.5\baselineskip]{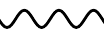} }
\newcommand*{\scalar}{ \includegraphics[height=0.5\baselineskip]{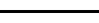} }
\newcommand*{\vertexQED}{ \raisebox{-0.3\baselineskip}{ \includegraphics[height=\baselineskip]{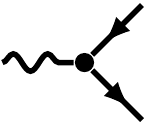} } }

\newcommand{\dslash}{\mathbin{
  \mathchoice{/\mkern-6mu/}% \displaystyle
    {/\mkern-6mu/}% \textstyle
    {/\mkern-5mu/}% \scriptstyle
    {/\mkern-5mu/}}}% \scriptscriptstyle

\newcommand{\Chi}{ \mathrm{X} }

\newcommand*{\defeq}{\mathrel{\vcenter{\baselineskip0.5ex \lineskiplimit0pt
                     \hbox{\scriptsize.}\hbox{\scriptsize.}}}%
                     =}

\DeclareMathOperator{\tr}{tr}

\makeatletter
\newsavebox{\@brx}
\newcommand{\llangle}[1][]{\savebox{\@brx}{\(\m@th{#1\langle}\)}%
	\mathopen{\copy\@brx\kern-0.5\wd\@brx\usebox{\@brx}}}
\newcommand{\rrangle}[1][]{\savebox{\@brx}{\(\m@th{#1\rangle}\)}%
	\mathclose{\copy\@brx\kern-0.5\wd\@brx\usebox{\@brx}}}
\makeatother

\newcommand{\al}[2][]{\begin{align#1}#2\end{align#1}}

\newtheorem{theo}{Theorem}[section]
\newtheorem{exam}[theo]{Example}
\newtheorem{lem}[theo]{Lemma}
\newtheorem{defi}[theo]{Definition}
\newtheorem{prop}[theo]{Proposition}

\newtheorem{rem}[theo]{Remark}

\newcommand{\reffig}[1]{fig. \ref{#1}}
\newcommand{\refeq}[1]{eq. (\ref{#1})}

\let\svthefootnote\thefootnote
\newcommand\blankfootnote[1]{%
  \let\thefootnote\relax\footnotetext{#1}%
  \let\thefootnote\svthefootnote%
}

\begin{document}

% TITLEPAGE

	% Metadata
	\title{\textbf{New graph polynomials in parametric QED Feynman integrals}}
	\date{}
	\author{Marcel Golz}
	\affil{Institut f\"ur Physik, Humboldt-Universit\"at zu Berlin\\ Newtonstraße 15, D-12489 Berlin, Germany}
	\maketitle
	\thispagestyle{empty}

	% Abstract
	\begin{abstract}
		In recent years enormous progress has been made in perturbative quantum field theory by applying methods of algebraic geometry to parametric Feynman integrals for scalar theories. The transition to gauge theories is complicated not only by the fact that their parametric integrand is much larger and more involved. It is, moreover, only implicitly given as the result of certain differential operators applied to the scalar integrand $\exp(-\frac{\Phi_{\Gamma}}{\Psi_{\Gamma}})$, where $\Psi_{\Gamma}$ and $\Phi_{\Gamma}$ are the Kirchhoff and Symanzik polynomials of the Feynman graph $\Gamma$. In the case of quantum electrodynamics we find that the full parametric integrand inherits a rich combinatorial structure from $\Psi_{\Gamma}$ and $\Phi_{\Gamma}$. In the end, it can be expressed explicitly as a sum over products of new types of graph polynomials which have a combinatoric interpretation via simple cycle subgraphs of $\Gamma$.
	\end{abstract}

	% table of contents	
	\tableofcontents

	% email footnote
	\blankfootnote{\texttt{mgolz@physik.hu-berlin.de}}

% MAIN PART
\newpage
% INTRODUCTION
	\section{Introduction}
	The parametric version of Feynman integrals has long been an extremely useful tool in the study of perturbative quantum field theory \cite{Nambu_1957_ParaRep, BergereZuber_1974_ParaRenorm, CvitanovicKinoshita_1974_FRPS, Panzer_2014_Algorithms, Panzer_2014_ManyScales}. Moreover, over the last decade a number of fascinating breakthroughs have unveiled deep connections to algebraic geometry and number theory, and have motivated mathematicians to study Feynman integrals, their periods, as well as connections to combinatorics and geometry \cite{BlochEsnaultKreimer_2006_Motives, Brown_2010_Periods, Brown_2009_Massless, BrownSchnetzYeats_2014_C2,BCDYeats_2015_ForbMinor, BrownSchnetz_2012_K3, AluffiMarcolli_2010_DetHyp}. However, most of this takes place in the realm of scalar quantum field theories due to the complications that the tensor structure of gauge theories brings to Feynman integrals. Not only does the parametric integrand in quantum electrodynamics (the simplest gauge theory) contain a number of (traces of) Dirac matrices, which we will discuss in a separate article \cite{Golz_2017_Traces}. It also contains a complicated rational function in the Schwinger parameters, momenta, metric tensors etc. in front of the usual integrand of a scalar Feynman integral. While the rough structure of this function has been known for a long time\footnote{It follows very directly from the Leibniz rule of differentiation. See also \cite{CvitanovicKinoshita_1974_FRPS}.} to be a certain sum over products of polynomials that are somehow related to derivatives of the second Symanzik polynomial, there has been no direct combinatorial interpretation of what these polynomials are. In this article we give such an interpretation of the form $\sum_C \pm \Psi_{\Gamma\dslash C}$, where the sum is over certain cycle subgraphs of the Feynman graph $\Gamma$ and $\Psi_{\Gamma\dslash C}$ is its Kirchhoff polynomial after contraction of that cycle.\\
	
	Moreover, we believe that combining the results of this article with those of \cite{Golz_2017_Traces} in future work will allow us to give a version of a parametric QED Feynman integrand that is essentially a single scalar integrand and includes a number of intricate cancellations that reduce the size of the integrand at higher loop orders by several orders of magnitude compared to the naive version \refeq{eq_QEDint}. A thus simplified Feynman integral will be much easier to handle, for example when trying to extend Brown's and Kreimer's parametric renormalization procedure \cite{BrownKreimer_2013_AnglesScales} (recently applied to great effect in \cite{KompanietsPanzer_2016}) to QED, and allow for a better understanding of non-scalar Feynman amplitudes. In particular, we see this as a first step in answering a number of long standing questions in QED, for example the cancellation of transcendental terms in the beta function \cite{Rosner_1966_QED, GKLS_1991_QED, BroadhurstDelbourgoKreimer_1996_Unknotting}. Finally, since most of the combinatorics underlying our result is independent of the specific case of QED, it should be possible to generalise the insights gained in this article to the non-abelian case by studying the corolla differential of \cite{KreimerSarsSuijlekom_2013_QuantGauge}.\\
	
	We begin by recapitulating some basic graph theory and the definitions of the Kirchhoff and Symanzik polynomials in sections \ref{subsec_graphs} and \ref{subsec_classpoly}. For more detail we suggest the reader consult the excellent review article \cite{BognerWeinzierl_2010_FeynmanPolynomials} or the classic book \cite{Nakanishi_1971_GraphFeynman}.  Building on that we define our new \textit{cycle polynomials} and discuss a number of examples and properties in section \ref{subsec_cyclepol}. In particular we would like to highlight the three identities proved in the lemmata \ref{lem_cyclepol_decomp}-\ref{lem_bondcycle}, since they are the fundamental building blocks for the proof of our main result and also quite fascinating in their own right. After briefly introducing parametric Feynman integrals for non-experts and discussing the peculiarities of quantum electrodynamics in section \ref{sec_FI} we are ready to state and prove theorem \ref{theo_main_numPoly}.

% GRAPHS AND GRAPH POLYNOMIALS
\section{Graphs and graph polynomials}

% Graphs, subgraphs and Feynman graphs
\subsection{Graphs, subgraphs and Feynman graphs}\label{subsec_graphs}

	% DEF graph, subgraph
	A \textit{graph} $G$ is an ordered pair $(V_G, E_G)$ of the set of \textit{vertices} $V_G = \{ v_1, \dotsc, v_{|V_G|}\}$ and the set of \textit{edges} $E_G = \{ e_1, \dotsc, e_{|E_G|}\}$, together with a map $\partial: E_G \to V_G \times V_G$, which is usually realised by drawing the graph in the plane. We will need our graphs to be directed, however as usual the particular choice of direction for each edge is arbitrary and will have no influence on the results of this article. For a directed edge $e\in E_G$ we write $\partial_-(e)\in V_G$ for its start vertex and $\partial_+(e)\in V_G$ for its target vertex, such that
	\al{
		\partial: e \mapsto ( \partial_-(e), \partial_+(e) )
	}	
	 Unless explicitly stated otherwise we assume $G$ to be connected, but its subgraphs may have multiple components. If a subgraph $g\subset G$ does not contain isolated vertices (which is the case for all the types of subgraphs we discuss below) it is uniquely defined by its edge set via $\partial(E_g)$ and we use the notation for the edge subset and the actual subgraph interchangeably. 

	% DEF important subgraphs
\paragraph{Types of subgraphs.}
	There are a multitude of significant types of graphs. For our purposes we concentrate on three of them, spanning trees, bonds and cycles.\\
	
	 A \textit{spanning tree} $T\subset G$ is a tree (i.e. a connected and simply connected graph) that contains all vertices of $G$. In other words:
	 \al[*]{
	 	h_0(T) = 1 \qquad h_1(T) = 0 \qquad V_T = V_G
	 }
	We denote the set of all spanning trees of $G$ by $\mathcal T_G$.\\
	
	 A \textit{bond} $B\subset G$ is a minimal subgraph such that $G' = G\setminus B  \equiv (V_G, E_G \setminus B)$ has exactly two connected components. We denote the set of all bonds of $G$ with $\mathcal B_G$.\\
	
	 A \textit{cycle} $C\subset G$ is a subgraph of $G$ that is 2-regular. In other words, a cycle is a disjoint union of closed paths in $G$ with no repeated vertices or edges.\footnote{In this we expressly include self-cycles (``tadpoles''), i.e. cycles with only one edge.} A cycle is called \textup{simple} if it consists of only one connected component. We denote the set of all cycles of $G$ with $\mathcal C_G$ and the set of simple cycles with $\mathcal C_G^{[1]}$.

	% REM vector space
	\begin{rem}
		Consider the vector space of edge subsets of a graph $G$ over $\mathbb Z_2$, where addition is given by the symmetric difference
		\al{
			E_1\triangle E_2 \defeq (E_1\setminus E_2) \cup (E_2\setminus E_1) = (E_1\cup E_2) \setminus (E_1\cap E_2)
		} 
		and the inner product is
		\al{
			\langle E_1, E_2\rangle \defeq \begin{cases}
								\ 1 \text{ if } |E_1\cap E_2| \text{ odd}\\
								\ 0 \text{ if } |E_1\cap E_2| \text{ even}.
								\end{cases}
		}
		$\mathcal C_G$ and $\mathcal B_G$ are each others orthogonal complement in this vector space and thus span it.  While we do not explicitly use it, this structure is the underlying mechanism behind our results and for example manifestly visible in \refeq{eq_con-del_bondcycle}.
	\end{rem}

	% DEF notions of orientation
\paragraph{Notions of orientation.}
	It will be useful to introduce some different notions of orientation into these graphs. One can assign an orientation to any simple cycle $C \in \mathcal C_G^{[1]}$ by specifying a direction in which to traverse it. The \textit{relative orientation} of an edge $e$ with respect to a simple cycle $C$ is
	\begin{align*}
		\mathrm o_{C}(e) \defeq \begin{cases} +1 \text{ if } e \text{ is traversed along its direction}\\  -1 \text{ if } e \text{ is traversed opposite its direction}\\ \ \ \, 0 \text{ if } e\notin C \end{cases}
	\end{align*}
	For more than one edge we simply generalise the notation by defining
	\begin{align*}
		\mathrm o_{C}(e_1,\dotsc,e_n) \defeq \mathrm o_{C}(e_1)\dotsm\mathrm o_{C}(e_n).
	\end{align*}
	This clearly depends on the initial choice of orientation for $C$ if $n$ is odd. For our purposes this is no problem since we only need either the case $n=2$ or $n=1$ within squared terms. Similarly one can define the orientation of a bond by choosing one of the two connected components in its complement $G\setminus B = G_1 \sqcup G_2$, such that
	\begin{align*}
		\mathrm o_{B}(e) \defeq \begin{cases} +1 \text{ if } e \text{ is directed into the chosen component,}\\  -1 \text{ if } e \text{ is directed out of the chosen component,}\\ \ \ \, 0 \text{ if } e\notin B. \end{cases}
	\end{align*}
	For multiple edges one can again use a similar notation as in the cycle case and the same remarks about dependence on the initial choice apply.\\
	
	% FIGURE subgraph examples
	\begin{figure}[h]\begin{center}
		\includegraphics[scale=1.0]{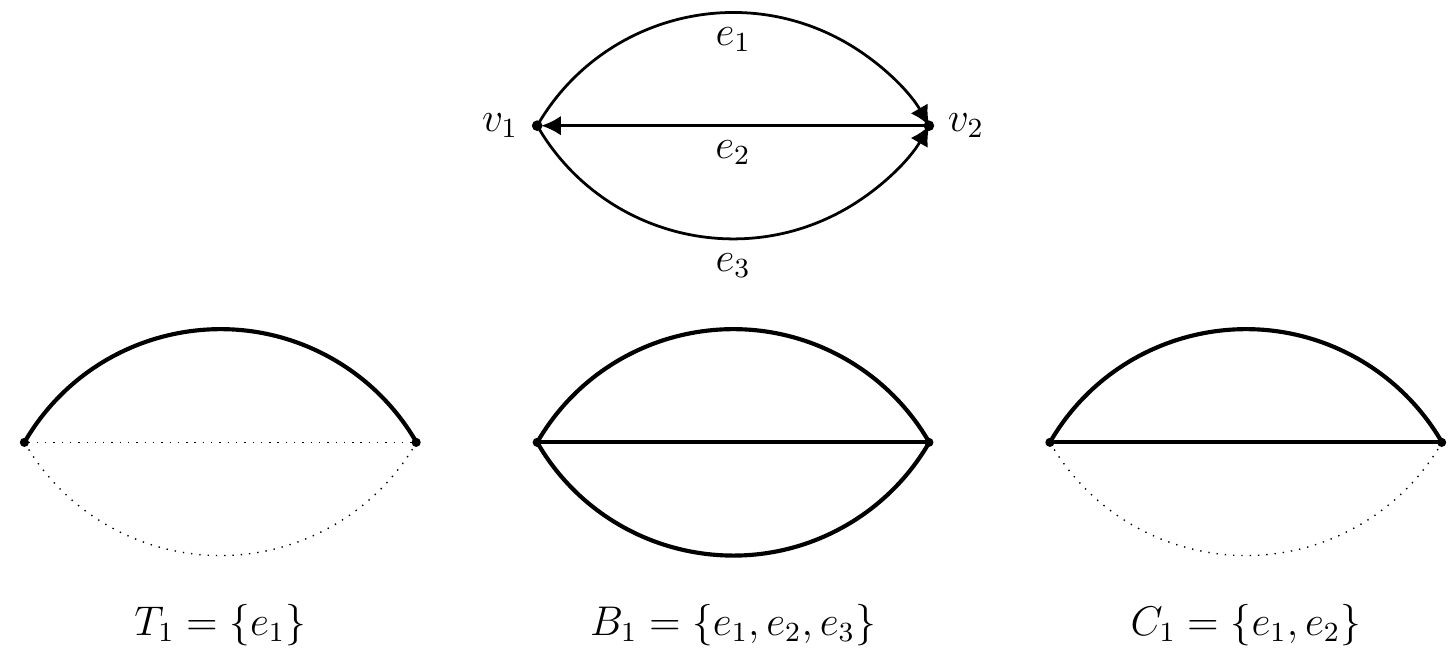}
		\caption[Spanning trees, bond, cycle]{The banana graph with 3 edges and examples for a spanning tree, a bond and a cycle subgraph of it.}\label{Example_01_Subgraphs}
	\end{center}\end{figure}

	% EXAMPLE subgraphs
	\begin{exam}
		Let $G$ be the banana graph with three edges depicted in \reffig{Example_01_Subgraphs}. It has three spanning trees
		\al[*]{
			T_1 = \{e_1\} \qquad T_2 = \{e_2\} \qquad T_3 = \{e_3\},
		}
		each consisting of a single edge. There is only one bond, $B_1 = \{e_1,e_2,e_3\}$, since all edges have to be removed to separate the graph into two components. Choose the orientation such that $\mathrm o_{B_1}(e_1) = +1$. Then $\mathrm o_{B_1}(e_3) = +1$, too, while $\mathrm o_{B_1}(e_2) = -1$. There are three cycles in $G$, each given by a pair of edges:
		\al[*]{
			C_1 = \{e_1,e_2\} \qquad C_2 = \{e_1,e_3\} \qquad C_3 = \{e_2,e_3\}
		}
		Choose an orientation for the cycles, say clockwise. Then the relative orientations of each edge with respect to each cycle are
		\al[*]{
			\mathrm o_{C_1}(e_1) &= +1\qquad\mathrm o_{C_1}(e_2) = +1\qquad\mathrm o_{C_1}(e_3) = 0\\
			\mathrm o_{C_2}(e_1) &= +1\qquad\mathrm o_{C_2}(e_2) = 0\hspace{33pt}\mathrm o_{C_2}(e_3) = -1\\
			\mathrm o_{C_3}(e_1) &= 0\hspace{33pt} \mathrm o_{C_3}(e_2) = -1\qquad\mathrm o_{C_3}(e_3) = -1
		}
	\end{exam}

	% DEF Feynman graphs	
\paragraph{Feynman graphs.} Feynman graphs are graphs with some extra information that can be used to encode the Feynman integrals we are interested in. There are two major significant differences compared to usual graphs. Firstly, there are so called \textit{external edges}, which are half-edges only incident to one vertex. They carry information about physical data like momenta of the incoming and outgoing particles they represent.  Secondly there can be different types of edges that represent different types of particles. Among others, there are \textit{scalars} (\scalar), \textit{photons} (\photon) and \textit{fermions} (\fermion). In the case of fermions we assume that the edge orientation is always given by the fermion flow direction (indicated by the arrow on the edge), for other edges it remains arbitrary. Feynman graphs often have restrictions on how many and which kind of edges are allowed to be incident to each vertex. For example in QED each vertex has to have one fermion oriented into it, one out of it and one photon line. Notationally we will use $G$ for general graphs and $\Gamma$ for Feynman graphs.
	
	% FIGURE Feynman graphs
	\begin{figure}[h] \begin{center}
		\includegraphics[scale=1.0]{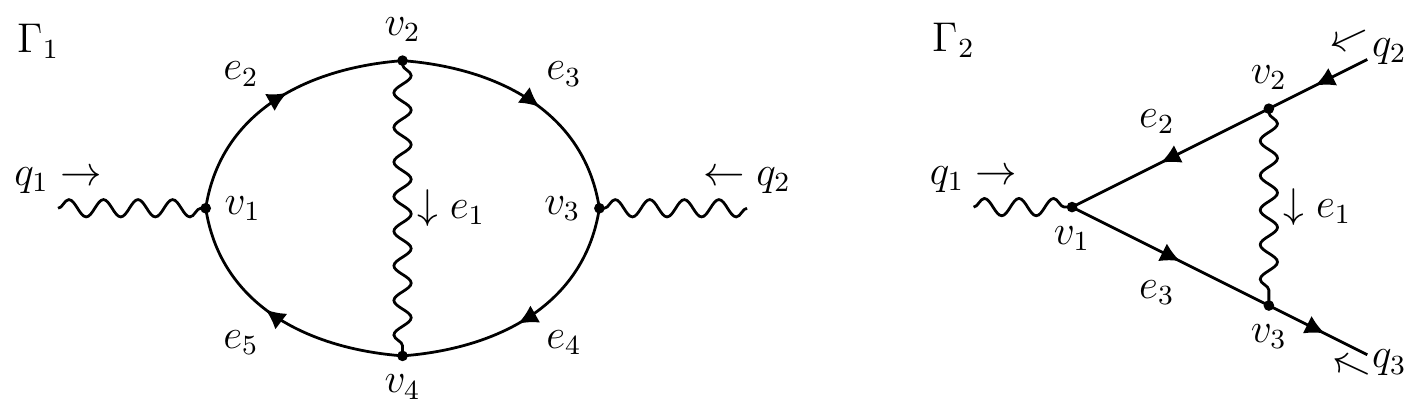}
		\caption[Feynman graphs]{Two examples of Feynman graphs from quantum electrodynamics. The $e_i$ and $v_i$ label edges and vertices while the $q_i$ are external momenta entering a graph via the external half-edges.}
		\label{Example_02_FeynmanGraphs}
	\end{center}\end{figure}

	% DEF classical graph polynomials
\subsection{Classical graph polynomials}\label{subsec_classpoly}
	Having introduced the three kinds of subgraphs that are important to us, we can assign to them polynomials that will appear in many places throughout. Each graph polynomial for a graph $G$ contains the variables $\alpha_e\in (0,\infty)\ \forall e \in E_G$ and possibly other formal parameters that we will introduce when needed. We will often use the abbreviations $\alpha = (\alpha_e)_{e \in E_G}$ and
	\al{
		\alpha_{S} \defeq \prod_{e \in S}\alpha_e
	}
	for any edge subset $S\subset E_G$.\\

	The first one is the \textit{Kirchhoff polynomial}, also often called first Symanzik polynomial, which corresponds to spanning trees. It is defined as
	\begin{align}
		\Psi_G(\alpha) \defeq  \sum_{T\in \mathcal T_G} \prod_{e\notin T} \alpha_e = \sum_{T\in \mathcal T_G} \alpha_{E_G \setminus T},\label{eq_def_Kirchhoff}
	\end{align}
	and has been known for a very long time, reaching back to Kirchhoff's study of electrical circuits \cite{Kirchhoff_1847_Kirchhoff}.\\

	The next polynomial is the \textit{second Symanzik polynomial}. Most commonly it is defined similarly to the Kirchhoff polynomial, by summing over spanning $2$-forests. Instead, we use an alternative definition\footnote{Equivalence can be seen quickly by noting that each spanning 2-forest lies in the complement of some bond, all spanning 2-forest in the same bond's complement share the same coefficient in the $\xi_e$ parameters, and their monomials $\alpha_{E_G \setminus (T_1 \cup T_2)}$ precisely make up $\alpha_B\Psi_{G\setminus B}$ for that bond.} via bonds that is, as we will see later, in many ways much more natural than the commonly used one:
	\begin{align}
		\Phi_G(\alpha, \xi) \defeq \sum_{B \in \mathcal B_G} \left(\sum_{e \in B} \mathrm o_{B}(e) \xi_e  \right)^2 \alpha_B \Psi_{G\setminus B}(\alpha)\label{eq_def_2ndSym}
	\end{align}
	Here, $\xi = (\xi_e)_{e \in E_G}$ is a tupel of formal parameters discussed in more detail below and $\Psi_{G\setminus B} = \Psi_{G_1(B)}\Psi_{G_2(B)}$ is the Kirchhoff polynomial of the bond's complement graph, with $G_1(B)$ and $G_2(B)$ denoting the two connected components.

	% REMARK evalutation physical Symanzik
	\begin{rem}\label{rem_evalSymanzik}
		In the case of Feynman graphs we can exchange the formal auxiliary parameters for physical momenta. Let $v_1,v_2$ be two external vertices, i.e. vertices with a half-edge incident to them, with incoming/outgoing momentum $q$. Choose any\footnote{The path independence of this is essentially Kirchhoff's voltage law, with momenta replacing voltages.} directed path between the vertices. Then replace the $\xi_e$ by
		\al[*]{
			\xi_e \to \begin{cases}
					\ \pm q \text{ if $e$ is in the path} \\
					\ \ \ 0 \text{ if $e$ is \textup{not} in the path.}
			\end{cases}
		}
		The sign of $q$ depends on the relative orientation of edge and path, similarly to what we defined above for cycles and bonds (in fact, each such path is of course just one of two segments of some cycle). This yields the usual physical second Symanzik polynomial. For example, the 3-edge banana from \reffig{Example_01_Subgraphs} has the second Symanzik polynomial $\Phi_G(\alpha,\xi) = (\xi_1-\xi_2+\xi_3)^2\alpha_1\alpha_2\alpha_3$. One can choose any single edge as the path, replace the corresponding $\xi_e$ by $q$ and set the other two to $0$ to get $\Phi_G(\alpha,q) = q^2\alpha_1\alpha_2\alpha_3$. If there are $n>2$ external vertices then one has to consider $n-1$ paths determined by momentum conservation, and replace $\xi_e$ by a sum of momenta, corresponding to all paths the edge is contained in. See also examples \ref{exam_classGraphPol}, \ref{ex_main_1} and \ref{ex_main_2}.
	\end{rem}

	% Disjoint union graph polynomials
	It is useful to extend the definition of graph polynomials to disjoint unions of graphs, and indeed we have already implicitly done that above in the definition of $\Phi_G$. Let $G = \bigsqcup_i G_i$ be a disjoint union of connected graphs $G_i$. Then the Kirchhoff polynomial of $G$ is simply the product of the polynomials of each of its components,
	\al{
		\Psi_G \defeq \prod_i \Psi_{G_i},
	}
	and the second Symanzik polynomial is
	\al{
		\Phi_G \defeq \sum_i \Phi_{G_i} \prod_{j\neq i} \Psi_{G_j}.
	}
	These formulae for the polynomials are precisely those for a vertex-1-connected graph that consists of the components $G_i$ arranged in a chain, each component overlapping with the next in only one vertex\footnote{See also \cite{BrownKreimer_2013_AnglesScales} and the ``circular joins'' used therein.}. It is sensible to define the polynomials for disjoint unions  $\bigsqcup_i G_i$ like this since there is clearly a one-to-one correspondence between spanning trees of such a vertex-1-connected graph and tuples of spanning trees, containing one tree from each component. Similarly, the bonds of such a graph are precisely the union of the bond sets of each component, hence the sum.
 	
	% Properties of graph polynomials.
\paragraph{Properties of graph polynomials.} These two polynomials have many interesting and useful properties. directly from the definition we observe that $\Psi_G$ and $\Phi_G$ are
	\begin{itemize}
		\item homogeneous of degree $h_1(G)$ and $h_1(G)+1$ in $\alpha_1,\dotsc,\alpha_{|E_G|}$,
		\item linear in each $\alpha_i$.
	\end{itemize}
	Moreover, they satisfy \textit{contraction-deletion} relations, which means that the polynomials belonging to graphs that are related via contraction or deletion of edges can be recovered easily from the original graph polynomial as follows\footnote{All this also holds for $\Phi_G$, except that the $\xi_e$ corresponding to the deleted or contracted edge also has to be set to 0.}:
	\al{
		\Psi_{G\dslash e}(\alpha) & = \left.\Psi_G(\alpha)\right|_{\alpha_e=0}\\[3mm]
		\Psi_{G\setminus e}(\alpha) & = \frac{\partial}{\partial \alpha_e} \Psi_G(\alpha)
	}
	From this one immediately finds a useful decomposition formula for graph polynomials:
	\al{
		\Psi_G(\alpha) = \Psi_{G\dslash e}(\alpha) + \alpha_e\Psi_{G\setminus e}(\alpha) \label{eq_con-del}
	}
	Note that two cases have to be excluded: One is the contraction of a tadpole edge, which is the same as just deleting the edge since its endpoints are already identified. The other is deletion of a bridge, which would disconnect a graph and assign to it a vanishing graph polynomial. This would of course be inconsistent with our extension of graph polynomials to graphs with more than one connected component.\\

	% EXAMPLE graph pols
	\begin{exam}\label{exam_classGraphPol}
		Let $\Gamma_1$, $\Gamma_2$ be the two Feynman graphs from \reffig{Example_02_FeynmanGraphs} above. Their Kirchhoff polynomials are
		\al{
			\Psi_{\Gamma_1}(\alpha) &= (\alpha_2+\alpha_5)(\alpha_3+\alpha_4) + \alpha_1(\alpha_2+\alpha_3+\alpha_4+\alpha_5)\\[5mm]
			\Psi_{\Gamma_2}(\alpha) &= \alpha_1+\alpha_2+\alpha_3,
		}
		while the second Symanziks in their general form are
		\al{\nonumber
			\Phi_{\Gamma_1}(\alpha,\xi) &= (\xi_2-\xi_5)^2\alpha_2\alpha_5(\alpha_1+\alpha_3+\alpha_4) +
								 (\xi_3-\xi_4)^2\alpha_3\alpha_4(\alpha_1+\alpha_2+\alpha_5)\\\nonumber
							&\quad\ +(\xi_1-\xi_2+\xi_4)^2\alpha_1\alpha_2\alpha_4 + 
									(\xi_1+\xi_3-\xi_5)^2\alpha_1\alpha_3\alpha_5\\
							&\quad\ +(\xi_1-\xi_2+\xi_3)^2\alpha_1\alpha_2\alpha_3 +
									(\xi_1+\xi_4-\xi_5)^2\alpha_1\alpha_4\alpha_5\\[5mm]
			\Phi_{\Gamma_2}(\alpha,\xi) &= (\xi_1+\xi_2)^2\alpha_1\alpha_2 + (\xi_1+\xi_3)^2\alpha_1\alpha_3 + (\xi_2 -\xi_3)^2\alpha_2\alpha_3.
		}
		All the properties mentioned above are clearly present. Furthermore, we show how to recover the physical versions $\Phi_{\Gamma_1}(\alpha,q_1,q_2)$ and $\Phi_{\Gamma_2}(\alpha,q_1,q_2,q_3)$ of the second Symanzik polynomials. For $\Gamma_1$ we need to choose any path between the vertices $v_1$ and $v_3$, say $\{e_2,e_3\}$. The external momenta entering these two vertices are $q_1$ and $q_2$, and by momentum conservation we know $-q_2 = q_1 \equiv q$. Consequently we have to evaluate $\xi_2,\xi_3 = q$ and $\xi_1,\xi_4,\xi_5 = 0$ and find
		\al{\nonumber
			\Phi_{\Gamma_1}(\alpha,q) &= (q-0)^2\alpha_2\alpha_5(\alpha_1+\alpha_3+\alpha_4) +
								 (q-0)^2\alpha_3\alpha_4(\alpha_1+\alpha_2+\alpha_5)\\\nonumber
							&\quad\ +(0-q+0)^2\alpha_1\alpha_2\alpha_4 + 
									(0+q-0)^2\alpha_1\alpha_3\alpha_5\\\nonumber
							&\quad\ +(0-q+q)^2\alpha_1\alpha_2\alpha_3 +
									(0+0-0)^2\alpha_1\alpha_4\alpha_5\\[1mm]
							& =  q^2\big( \alpha_2\alpha_5(\alpha_1+\alpha_3+\alpha_4)
							+\alpha_3\alpha_4(\alpha_1+\alpha_2+\alpha_5)+ \alpha_1\alpha_2\alpha_4+\alpha_1\alpha_3\alpha_5\big).
		}
		For $\Gamma_2$ one has to consider two different paths. Due to momentum conservation one has $q_3 = -q_1 - q_2$, and we need paths from $v_1$ and $v_2$ to $v_3$, say the single edges $e_1$ and $e_3$. The path $e_1$ corresponds to momentum flow of $q_2$ from $v_1$ to $v_3$, so one replaces $\xi_1 \to q_2$. Similarly for the other path one replaces $\xi_3 \to q_1$ and $\xi_2 \to 0$. Hence,
		\al{\nonumber
			\Phi_{\Gamma_2}(\alpha,q_1,q_2) &= (q_2+0)^2\alpha_1\alpha_2 + (q_2+q_1)^2\alpha_1\alpha_3 + (0 -q_1)^2\alpha_1\alpha_3\\[3mm]
								& = q_2^2\alpha_1\alpha_2 + (q_1+q_2)^2\alpha_1\alpha_3 + q_1^2\alpha_2\alpha_3.
		}
	\end{exam}

% Cycle polynomials
\subsection{The cycle polynomials}\label{subsec_cyclepol}

	\begin{defi}\textup{\textbf{(Cycle polynomials)}}\label{def_cycle_polynomial}\\
		Let $G$ be a connected graph, $\mathcal C_G$ its set of cycles and $\mathcal C_G^{[1]}$ the simple cycles. Then we define the cycle polynomial of $G$ to be
		\al{
			\chi_G(\alpha,\xi) \defeq \sum_{C \in \mathcal C_G^{[1]}} \left(\sum_{e \in C} \mathrm o_C(e) \xi_e  \right)^2   \Psi_{G\dslash C}(\alpha).
		}
		Furthermore we define two versions of the cycle polynomial restricted to certain subsets of $\mathcal C_G^{[1]}$. These will be the most important for our applications. For each $e\in E_G$ define
		\al{
			\chi^{(e)}_G(\alpha) \defeq \frac{1}{2}\frac{\partial^2}{\partial\xi_e^2} \chi_G(\alpha,\xi) = \sum_{C \in \mathcal C_G^{[1]}} \mathrm o_C(e)^2 \Psi_{G\dslash C}(\alpha) = \sum_{\substack{C \in \mathcal C_G^{[1]}\\ C \ni e}} \Psi_{G\dslash C}(\alpha).
		}
		and for each pair of distinct edges $e_1,e_2 \in E_G$ define
		\al{
			\chi^{(e_1,e_2)}_G(\alpha) \defeq \frac{1}{2}\frac{\partial^2}{\partial\xi_{e_1}\partial\xi_{e_2}} \chi_G(\alpha,\xi) = \sum_{C \in \mathcal C_G^{[1]}} \mathrm o_{C}(e_1,e_2) \Psi_{G\dslash C}(\alpha).
		}
	\end{defi}

	% REMARK definition/notation
	\noindent Furthermore, we define an expression $\chi^{(e,e)}_G$ for a pair in which both edges are the same:
	\al{
		\chi^{(e,e)}_G(\alpha,\varepsilon) \defeq \chi^{(e)}_G(\alpha) + \frac{2\Psi_G(\alpha)}{\varepsilon\alpha_e}\label{eq_chi_sameedge}
	}
	While this is not a polynomial anymore, it provides a notation that will allow us to write down our result very elegantly. In fact, we will see that the parameter $\varepsilon>0$ will be the gauge parameter of quantum electrodynamics.\\

	% REMARK cycle pol disjoint
	By using the definition of the Kirchhoff polynomial for disconnected graphs, the cycle polynomial definition above clearly extends to disconnected graphs as well, without any need for changes in notation. Similarly, all the properties and identities discussed below hold quite obviously in general, but we restrict ourselves to connected graphs $G$ for simplicity.

	% REMARK bond polynomial
	\begin{rem}\label{rem_defbondpol}
		The similarity of $\chi_G(\alpha,\xi)$ and the second Symanzik polynomial is no accident. In fact, one might consider the second Symanzik polynomial as a \textup{``bond polynomial''} $\beta_G(\alpha,\xi) \equiv \Phi_G(\alpha,\xi)$ and define $\beta^{(e)}_G(\alpha)$, $\beta^{(e_1,e_2)}_G(\alpha)$ analogously to what we did for cycles:
		\al{
			\beta_G^{(e)}(\alpha)&\defeq \frac{1}{2}\frac{\partial^2}{\partial\xi_e^2} \beta_G(\alpha,\xi) 
						= \sum_{B \in \mathcal B_G} \mathrm o_{B}(e)^2\alpha_B \Psi_{G\setminus B}(\alpha)
						= \sum_{\substack{B \in \mathcal B_G\\ B \ni e}} \alpha_B \Psi_{G\setminus B}(\alpha)\\[3mm]
			\beta^{(e_1,e_2)}_G(\alpha)&\defeq \frac{1}{2}\frac{\partial^2}{\partial\xi_{e_1}\partial\xi_{e_2}} \beta_G(\alpha,\xi)
							= \sum_{B \in \mathcal B_G} \mathrm o_B(e_1,e_2) \alpha_B \Psi_{G\setminus B}(\alpha)
		}
	\end{rem}

	% FIGURE contraction of cycles
	\begin{figure}[h] \begin{center}
		\includegraphics[scale=1.0]{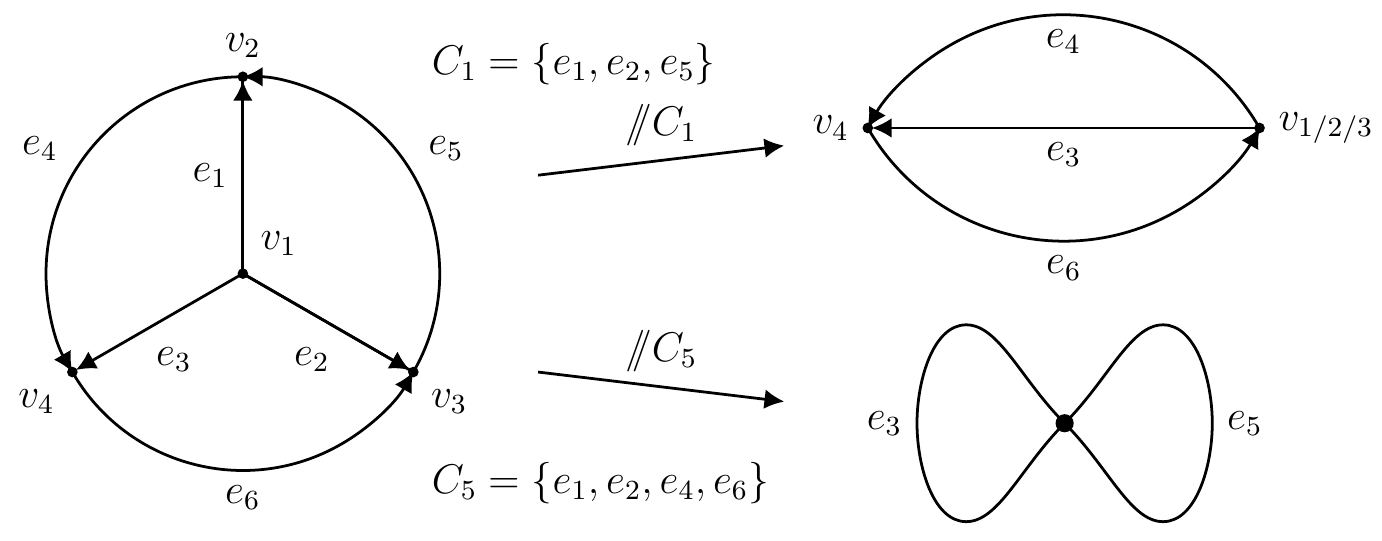}
		\caption[Cycle contraction in $WS_3$]{The \textit{wheel with 3 spokes} graph $WS_3$ and the contraction of two of its cycles.}
		\label{Example_03_CycleContraction}
	\end{center}\end{figure}

	% EXAMPLE cycle pol
	\begin{exam}
		Let $G = WS_3$ be the wheel with three spokes as depicted in \reffig{Example_03_CycleContraction}. It contains seven cycles, all of which are simple:
		\al[*]{
			C_1 & = \{e_1, e_2, e_5\} \qquad C_2 = \{e_1, e_3, e_4\} \qquad C_3 = \{e_2, e_3, e_6\} \qquad C_4 = \{e_4, e_5, e_6\}\\[3mm]
			C_5 & = \{e_1, e_2, e_4, e_6\} \hspace{19mm} C_6 = \{e_1, e_3, e_5, e_6\} \hspace{19mm} C_7 = \{e_2, e_3, e_4, e_5\}
		}
		Contracting the 3-edge cycles $C_1, C_2, C_3, C_4$ results in a 3-edge banana graph, while contraction of the 4-edge cycles returns a rose with the two remaining edges. Hence,
		\al[*]{
			\Psi_{G\dslash C_1}(\alpha) &= \alpha_3\alpha_4 + \alpha_3\alpha_6 + \alpha_4\alpha_6\qquad
			\Psi_{G\dslash C_2}(\alpha) = \alpha_2\alpha_5 + \alpha_2\alpha_6 + \alpha_5\alpha_6\\[3mm]
			\Psi_{G\dslash C_3}(\alpha) &= \alpha_1\alpha_4 + \alpha_1\alpha_5 + \alpha_4\alpha_5\qquad
			\Psi_{G\dslash C_4}(\alpha) = \alpha_1\alpha_2 + \alpha_1\alpha_3 + \alpha_2\alpha_3\\[3mm]
			\Psi_{G\dslash C_5}(\alpha) &= \alpha_3\alpha_5 \hspace{16mm} 
			\Psi_{G\dslash C_6}(\alpha) = \alpha_2\alpha_4 \hspace{16mm}
			\Psi_{G\dslash C_7}(\alpha) = \alpha_1\alpha_6
		}
		and the full cycle polynomial is
		\al{\nonumber
			&\hspace{-5mm}\chi_G(\alpha, \xi)=\\\nonumber
				& (\xi_1-\xi_2-\xi_5)^2\alpha_3\alpha_4 + \alpha_3\alpha_6 + \alpha_4\alpha_6
						+ (\xi_1-\xi_3+\xi_4)^2\alpha_2\alpha_5 + \alpha_2\alpha_6 + \alpha_5\alpha_6\\\nonumber
					      +& (\xi_2-\xi_3-\xi_6)^2\alpha_1\alpha_4 + \alpha_1\alpha_5 + \alpha_4\alpha_5
					      + (\xi_4+\xi_5+\xi_6)^2\alpha_1\alpha_2 + \alpha_1\alpha_3 + \alpha_2\alpha_3\\
					      +& (\xi_1-\xi_2+\xi_4+\xi_6)^2\alpha_3\alpha_5
  					        + (\xi_1-\xi_3-\xi_5-\xi_6)^2\alpha_2\alpha_4
					        + (\xi_2-\xi_3+\xi_4+\xi_5)^2\alpha_1\alpha_6
		}
		The restricted version, say for $e=e_1$ contains only the terms corresponding to cycles that contain $e_1$, so $C_1$, $C_2$, $C_5$ and $C_6$. Hence,
		\al{
			\chi_G^{(e_1)}(\alpha) &= (\alpha_2+\alpha_3)(\alpha_4+\alpha_5) + \alpha_6(\alpha_2 + \alpha_3 + \alpha_4 + \alpha_5).
		}
		Finally, consider two edges, say $e_1,e_2$. The only cycles that contain both are $C_1$ and $C_5$, and $e_1,e_2$ have opposing directions in both those cycles such that $\mathrm o_{C_1}(e_1,e_2) = -1 = \mathrm o_{C_5}(e_1,e_2)$. Hence
		\al{
			\chi_G^{(e_1,e_2)}(\alpha) &= - \alpha_3(\alpha_4 + \alpha_5 + \alpha_6) - \alpha_4\alpha_6.
		}
		For a different example consider $e_1,e_6$. There are again two cycles containing this pair, $C_5$ and $C_6$. However, this time $\mathrm o_{C_5}(e_1,e_6) = +1$ and $\mathrm o_{C_6}(e_1,e_6) = -1$, so
		\al{
			\chi_G^{(e_1,e_6)}(\alpha) &= \alpha_3\alpha_5 - \alpha_2\alpha_4.
		}
	\end{exam}

% Properties cycle pol
\paragraph{Properties of cycle polynomials.}
	
	Cycle polynomials have a number of interesting properties. One immediately notices that they inherit linearity in each $\alpha_e$ as well as homogeneity of degree $h_1(G)-1$ from the Kirchhoff polynomials. They also satisfy the usual contraction-deletion relations, which we prove in
	\begin{prop}
		Let $G$ be a connected graph. Then	 for every $e\in E_G$
		\al{
			\chi_{G\setminus e}(\alpha,\xi) = \frac{\partial}{\partial \alpha_e}\left. \chi_{G}(\alpha,\xi)\right|_{\xi_e=0},
		}
		and for every $e\in E_G$ except tadpoles
		\al{
			\chi_{G\dslash e}(\alpha,\xi) =  \left.\chi_{G}(\alpha,\xi)\right|_{\alpha_e,\xi_e=0}.
		}
		The same relations hold for the restricted versions $\chi_{G}^{(e_0)}(\alpha)$ and $\chi_{G}^{(e_1,e_2)}(\alpha)$.
	\begin{proof}
		Firstly, the cycle set of $G$ can be separated into two disjoint subsets, depending on whether or not the cycles contain a given edge $e$ and $\mathcal C_{G\setminus e}^{[1]} = \{ C \in \mathcal C_G^{[1]} \ | \ e \notin C \}$. Secondly, if a cycle $C$ contains an edge $e$, then $\Psi_{G\dslash C}$ is independent of $\alpha_e$ and the derivative w.r.t. $\alpha_e$ vanishes. For the cycles not containing $C$ we can employ the contraction-deletion relation for the Kirchhoff polynomial and commutativity of contracting and deleting to find
		\al{\nonumber
			\frac{\partial}{\partial \alpha_e}\left.\chi_{G}(\alpha)\right|_{\xi_e=0}
				& = \sum_{C \in \mathcal C_G^{[1]}} \Bigg(\sum_{\substack{e' \in C\\ e'\neq e}} \mathrm o_C(e') \xi_e'  \Bigg)^2 \frac{\partial}{\partial \alpha_e}\Psi_{G\dslash C}(\alpha)\\\nonumber
				&	= \sum_{\substack{C \in \mathcal C_G^{[1]}\\ e \notin C}} \Bigg(\sum_{\substack{e' \in C\\ e'\neq e}} \mathrm o_C(e') \xi_e'  \Bigg)^2\Psi_{(G\dslash C)\setminus e}(\alpha)\\
				&	= \sum_{C \in \mathcal C_{G\setminus e}^{[1]}} \Bigg(\sum_{e' \in C} \mathrm o_C(e') \xi_e'  \Bigg)^2 \Psi_{(G\setminus e)\dslash C}(\alpha)
					= \chi_{G\setminus e}(\alpha).
		}		
		For the contraction we need to consider three different cases. let $C \in \mathcal C_G^{[1]}$ and $e\in E_G$ any non-tadpole edge. If $e\in C$ then one observes that after contraction $C\dslash e$ is still a cycle of $G\dslash e$ and their corresponding Kirchhoff polynomials are unchanged since they are independent of $\alpha_e$. If $e\notin C$, but both endpoints of $e$ lie in $C$, then two things happen. On the one hand, contracting $e$ joins four edges of $C$ in one vertex such that $C$ is \textit{not} a cycle of $G\dslash e$ anymore. On the other hand, contracting $C$ turns $e$ into tadpole, such that $\Psi_{G\dslash C}(\alpha)$ vanishes when evaluating at $\alpha_e=0$. Therefore, these terms vanish on both sides of the contraction relation. Finally, when $e\notin C$ and at least one of its endpoints is not incident to a vertex of $C$, then clearly $C$ remains a cycle in $G\dslash e$ and
		\al{
			\left.\Psi_{G\dslash C}(\alpha)\right|_{\alpha_e=0} = \Psi_{(G\dslash C)\dslash e}(\alpha) = \Psi_{(G\dslash e)\dslash C}(\alpha).
		}
		Repeating the same steps restricted to cycle subsets containing $e_0$ or $e_1$, $e_2$ proves the deletion relation for the restricted polynomials.
	\end{proof}
	\end{prop}
	
	Additionally, there are relations between $\Psi_{G\setminus e}$, $\Psi_{G\dslash e}$ and $\chi_G^{(e)}$, as well as the bond polynomial $\beta_G^{(e)}$ defined in remark \ref{rem_defbondpol}. The identities established in the following three lemmata already contain most of the work needed to prove our main result in section \ref{Sec_MainResult}. In particular, the first two lemmata give us an alternative version of the contraction-deletion relation \refeq{eq_con-del} in terms of the cycle and bond polynomial
	\al{
		\Psi_G(\alpha) = \Psi_{G\dslash e}(\alpha) + \alpha_e\Psi_{G\setminus e}(\alpha) = \alpha_e^{-1}\beta^{(e)}_G(\alpha) + \alpha_e\chi^{(e)}_G(\alpha),
		\label{eq_con-del_bondcycle}
	}
	while the third introduces a relation between cycles and bonds.
	
% LEMMA cycle = delete
	\begin{lem}\label{lem_cyclepol_decomp}
		Let $G$ be a connected graph and $e\in E_G$ any of its edges that is not a bridge\footnote{If it were a bridge then no cycle would contain it such that $\chi_G^{(e)}(\alpha) = 0$. This would be inconsistent with our multiplicative definition of graph polynomials for graphs with multiple connected components that demands $\Psi_{G\setminus e}(\alpha) \neq 0$.}. Then
		\al{
			\chi_G^{(e)}(\alpha) = \Psi_{G\setminus e}(\alpha).
		}
	\begin{proof}
		The left hand side written out in full is
		\al{
			\chi^{(e)}_G(\alpha) = \sum_{\substack{C \in \mathcal C_G^{[1]}\\ C \ni e}} \Psi_{G\dslash C}(\alpha) = \sum_{\substack{C \in \mathcal C_G^{[1]}\\ C \ni e}} \sum_{T \in \mathcal T_{G\dslash C}} \alpha_{E_G\setminus (C\cup T)}.
		}	
		By definition all cycles summed over here contain $e$, and contracting a cycle $C$ is the same as contracting all but one of its edges and deleting the last edge (which is only a tadpole after contraction of all the others and thus not contained in any spanning trees). Hence, we can replace the sum over cycles by a sum over paths $P$ between vertices $v$, $w$ in $G\setminus e$, where $(v, w) = \partial(e)$, i.e.
		\al{
			\chi^{(e)}_G(\alpha) = \sum_{P \in \mathcal P^{v,w}_{G\setminus e}}\Psi_{(G\setminus e)\dslash P}(\alpha).
		}
		Two observations suffice to finish the proof. Firstly, a spanning tree remains a spanning tree after contraction of any number of its edges, i.e. for a connected graph $H$ one of its spanning trees $T$ and any edge subset $T'\subseteq T$ it follows that $T\dslash T'$ is also a spanning tree of $H\dslash T'$. Secondly, \textit{every} spanning tree contains a \textit{unique} path between any two vertices of the graph it spans. Therefore
		\al{
			\mathcal T_{(G\setminus e)\dslash P} \simeq \{ T\in \mathcal T_{G\setminus e} \ | \ T \supset P \}		
		}
		and
		\al{
			\bigcup_{P \in \mathcal P^{v,w}_{G\setminus e}} \mathcal T_{(G\setminus e)\dslash P} \simeq \mathcal T_{G\setminus e}, \qquad \bigcap_{P \in \mathcal P^{v,w}_{G\setminus e}} \mathcal T_{(G\setminus e)\dslash P} = \emptyset.
		}
		Finally, for the monomial it is irrelevant whether an edge is contained in the spanning tree or was contracted - each monomial contains variables associated to edges that are still in the graph and not in the spanning tree. In other words, for spanning trees $T_1\in\mathcal T_{(G\setminus e)\dslash P}$ and $T_2 \in \mathcal T_{G\setminus e}$ with $P\subseteq T_2$ one has $\alpha_{E_G\setminus ( e \cup P \cup T_1 )}  = \alpha_{E_G\setminus ( e \cup T_2 )}$ such that
		\al{
			\chi^{(e)}_G(\alpha) = \sum_{P \in \mathcal P^{v,w}_{G\setminus e}}\sum_{T \in \mathcal T_{(G\setminus e)\dslash P}} \alpha_{E_G\setminus (e \cup P \cup T)} = \sum_{T \in \mathcal T_{G\setminus e}} \alpha_{E_G\setminus (e \cup T)} = \Psi_{G\setminus e}(\alpha).
		}
	\end{proof}
	\end{lem}

% LEMMA Bond = contract
	\begin{lem}\label{lem_bondpol_decomp}
		Let $G$ be a connected graph and $e\in E_G$ any of its edges that is not a tadpole. Then
		\al{
			\beta_G^{(e)}(\alpha) = \alpha_e \Psi_{G\dslash e}(\alpha).
		}
	\begin{proof}
		Let $B \in \mathcal B_G$ be a bond that contains $e$ and  $G'(B)$ the graph obtained from $G\setminus B$ by identifying the vertices that $e$ is incident to, i.e.
		\al{
			G'(B) = (G\setminus (B\setminus e))\dslash e.
		}
		 Then $\mathcal T_{G\setminus B} \simeq \mathcal T_{G'(B)}$ and $\Psi_{G\setminus B}(\alpha) = \Psi_{G'(B)}(\alpha)$. Clearly, $\mathcal T_{G'(B)} \subset \mathcal T_{G\dslash e}$. Consider any spanning tree $T \in \mathcal T_{G\dslash e}$. The edges of $T$ form a spanning 2-forest in $G$ and their complement in $G$ contains \textit{exactly one} of the bonds $B \in \mathcal B_G$, given by the edges that are bridges between the two connected components. Therefore,
	\al{
		\bigcup_{\substack{B \in \mathcal B_G\\ B\ni e\ \ }} \mathcal T_{G'(B)} = \mathcal T_{G\dslash e} \quad\text{ and }\quad \bigcap_{\substack{B \in \mathcal B_G\\ B\ni e\ \ }} \mathcal T_{G'(B)} = \emptyset.
	}
	Moreover, for each monomial one sees that
	\al{
		\alpha_B \alpha_{E_G\setminus (T\cup B)} = \alpha_{E_G \setminus T} = \alpha_e \alpha_{E_G \setminus (e\cup T)}
	}
	from which we conclude that
	\al{
		\beta_G^{(e)}(\alpha) 
		= \sum_{\substack{B \in \mathcal B_G\\B\ni e}} \alpha_B \sum_{T \in \mathcal T_{G'(B)}} \alpha_{E_G\setminus (T\cup B)} 
		= \sum_{T \in \mathcal T_{G\dslash e}} \alpha_{E_G\setminus T} =  \alpha_e \Psi_{G\dslash e}(\alpha).
	}
	\end{proof}
	\end{lem}

% LEMMA bond / cycle relation
	\begin{lem}\label{lem_bondcycle}
		Let $G$ be a connected graph and $e,e' \in E_G$ two distinct edges. Then
		\al{
			\beta_G^{(e,e')}(\alpha) = -\alpha_e\alpha_{e'} \chi_G^{(e,e')}(\alpha).
		}
	\begin{proof}
		By definition
		\al{
			\beta_G^{(e,e')}(\alpha)  =  \sum_{B \in \mathcal B_G}  \mathrm o_{B}(e,e') \alpha_B\Psi_{G\setminus B}(\alpha).
		}
		 Consider the two connected components $G_1,G_2$ of $G\setminus B$ and remember that $\Psi_{G\setminus B}(\alpha) = \Psi_{G_1}(\alpha)\Psi_{G_2}(\alpha)$. In each component let $v_i, v_i'\in V_{G_i}$ be the two (not necessarily distinct) vertices that $e$ and $e'$ are incident to respectively. By the arguments of the proof of lemma \ref{lem_cyclepol_decomp} we can write
		\al{
			\Psi_{G_i}(\alpha)  =  \sum_{P_i} \Psi_{G_i\dslash P_i}(\alpha).
		}
		where the sum is over all paths $P_i\in \mathcal P_{G_i}^{v_i,v_i'}$ between the two vertices. After contraction of any two paths $P_1$, $P_2$ the edges $e$ and $e'$ start and end in the same two vertices, such that we get another isomorphism on spanning trees.
		\al{
			\mathcal T_{G_1\dslash P_1}\times\mathcal T_{G_2\dslash P_2} \simeq \mathcal T_{(G\setminus B')\dslash C_0}
		}
		where $B'= B \setminus \{e,e'\}$ and $C_0$ is the simple cycle formed by the two paths and $e$, $e'$. Clearly, every simple cycle of $G\setminus B'$ that contains $e$ and $e'$ can be constructed by combining the two edges with any pair of paths $P_1$, $P_2$, i.e.
		\al[*]{
			\{ P_1 \} \oplus \{ P_2 \} \oplus e \oplus  e' \simeq \{ C \in \mathcal C_{G\setminus B'}^{[1]} \ | \ C \supseteq \{e,e'\}\} = \{ C \in \mathcal C_G^{[1]} \ | C \cap B = \{e,e'\} \}.
		}
		On the level of the Kirchhoff polynomial we can therefore write
		\al{
			\Psi_{G\setminus B}(\alpha)  = \sum_{P_1}\sum_{P_2}  \Psi_{G_1\dslash P_1}(\alpha)\Psi_{G_2\dslash P_2}(\alpha) =  \sum_{\substack{C \in \mathcal C_G^{[1]}\\ C\cap B = \{e,e'\}}} \Psi_{(G\setminus B')\dslash C }(\alpha).
		}
		We also see from this construction that for any bond $B$ and simple cycle $C$ that intersect exactly in these two edges the relative orientations of $e,e'$ are related via $\mathrm o_B(e,e') = -\mathrm o_C(e,e')$. Furthermore, contraction and deletion of edge subsets commute as long as the contracted and deleted set do not intersect. We can therefore exchange summation over bonds and cycles as follows:
		\al{\nonumber
			\sum_{B \in \mathcal B_G} \mathrm o_{B}(e,e') \alpha_B \Psi_{G\setminus B}(\alpha)
			& = \sum_{B \in \mathcal B_{G}} \mathrm o_{B}(e,e') \alpha_B \sum_{\substack{C \in \mathcal C_G^{[1]}\\ C\cap B = \{e,e'\}}} \Psi_{(G\setminus B')\dslash C }(\alpha)\\
			& = -\sum_{C \in \mathcal C_G^{[1]}} \mathrm o_C(e,e') \sum_{\substack{B \in \mathcal B_G\\ B \ni e,e'}}\alpha_B \Psi_{(G\dslash C)\setminus B' }(\alpha).
		}
		Note also that in the sum over bonds we can weaken the requirement $C\cap B = \{e,e'\}$ to $e,e' \in B$ since removing edges that have previously been contracted would yield a vanishing Kirchhoff polynomial and thus give no contribution anyway. Let now $G'$ be the graph $G$ in which the edges $C\setminus \{e,e'\}$ have been contracted and $e'$ deleted. Then clearly $G\dslash C = G'\dslash e$ and we can apply lemma \ref{lem_bondpol_decomp} to $G'$ to get
		\al{
			\sum_{B \in \mathcal B_G} \mathrm o_B(e,e') \alpha_B \Psi_{G\setminus B}(\alpha)
			= -\alpha_e\alpha_{e'}\sum_{C \in \mathcal C_G^{[1]}} \mathrm o_C(e,e') \Psi_{G\dslash C}(\alpha)
			= -\alpha_e\alpha_{e'} \chi_G^{(e,e')}(\alpha).
		}
	\end{proof}
	\end{lem}

	Last but not least we prove another interesting identity involving our cycle polynomials. While one would assume such a rather simple identity must have been known to people working extensively with parametric Feynman integrals in the past, we are not aware of it being used or proved in any published article.

% PROPOSITION Kirchhoff/Cycle
	\begin{prop}
		Let $G$ be a graph. Then
		\al{
			\Psi_G(\alpha) = \frac{1}{h_1(G)}\sum_{C_\in\mathcal C_G^{[1]}}\Psi_C(\alpha)\Psi_{G\dslash C}(\alpha)
		}
		and 
		\al{
	 		\Psi_G(\alpha) = \frac{1}{h_1(G)}\sum_{e\in E_G} \alpha_e \chi_G^{(e)}(\alpha).
		}
	\begin{proof}
		One can quickly see equality of the two right-hand sides by noting that
		\al{
			\Psi_C(\alpha) = \sum_{e\in C} \alpha_e
		}
		and exchanging summations. In order to prove equality with the Kirchhoff polynomial consider first that we have seen in lemma \ref{lem_cyclepol_decomp} that $\alpha_e \chi_G^{(e)}(\alpha)$ is part of the Kirchhoff polynomial, i.e. each monomial on the right-hand sides does appear in $\Psi_G(\alpha)$. On the other hand, consider any monomial of $\Psi_G(\alpha)$. It is a product of parameters $\alpha_e$ corresponding to $h_1(G)$ distinct edges, none of which is a bridge since such an edge would be contained in every spanning tree. Since they are not bridges, each edge $e$ is contained in at least one simple cycle $C_e$ such that the monomial is contained in $\Psi_{C_e}(\alpha)\Psi_{G\dslash C_e}(\alpha)$ and appears at least $h_1(G)$ times on the right-hand side. Furthermore, for two distinct simple cycles $C,C'$ that both contain $e$ the polynomials $\Psi_{G\dslash C}(\alpha)$ and $\Psi_{G\dslash C'}(\alpha)$ share no monomials since otherwise by lemma \ref{lem_cyclepol_decomp} $\Psi_G(\alpha)$ would contain monomials with coefficient not in $\{0,1\}$. Therefore, each monomial of $\Psi_G(\alpha)$ indeed appears exactly $h_1(G)$ times in the sums on the right-hand side.
	\end{proof}
	\end{prop}

% PARAMETRIC FEYNMAN INTEGRALS
\section{Parametric Feynman integrals}\label{sec_FI}

Here we first briefly introduce the idea of parametric Feynman integrals for readers unfamiliar with them and in the second subsection state and explain the problem to be solved by our main result in the following section.

% Perturbative QFT and Schwinger parameters
\subsection{Perturbative QFT and Schwinger parameters}

	There are a number of issues that we will not discuss at all here. Firstly, Feynman integrals are generally divergent and have to undergo \textit{renormalisation} (essentially a certain sequence of blow-ups, from a mathematical viewpoint \cite{BrownKreimer_2013_AnglesScales}) to yield sensible finite results. In this article we are only interested in the integrands and do not wish to integrate yet, so this is of no concern to us. In order to simplify further we also only consider massless Euclidean integrands, but our results should not significantly differ in the massive or Minkowskian case\footnote{With masses one would need to include a summand $-\sum_{e\in E_{\Gamma}} m_e^2\alpha_e$ in the exponential that appears in the integrand. While this complicates integration considerably it is merely a multiplicative constant for the purposes of this article. Minkowskian and Euclidean cases can be related via the standard methods of Wick rotation.}. Finally, we work with the Schwinger parametric version of Feynman integrals only and discuss neither momentum or position space, nor how to derive one form of Feynman integrals from the others in much detail.\\

	In perturbative quantum field theory (pQFT) one considers a series expansion 
	\al{
		\mathcal A = \mathcal A^{(0)} + g\mathcal A^{(1)}+ g^2\mathcal A^{(2)}+\dotsm
	}
	of probability amplitudes. The coefficients are given by a sum of so called \textit{Feynman integrals} that can be encoded by a kind of graph, likewise named after Feynman, and the $i$-th coefficient contains the Feynman graphs with first Betti number $i$ and the suitable number and types of external half edges. In the simplest case of a scalar QFT in four space time dimensions the integrand of a Feynman integral in its parametric form is (up to trivial factors) given by\footnote{Note that when using the graph polynomials in formulae for Feynman integrals we will omit the arguments to reduce clutter in the notation.}\cite{Nambu_1957_ParaRep,BognerWeinzierl_2010_FeynmanPolynomials}
	\al{
		I_{\Gamma} \defeq \frac{e^{-\frac{\Phi_{\Gamma}}{\Psi_{\Gamma}}}}{\Psi_{\Gamma}^2},
	}
	where $\Gamma$ is a scalar Feynman graph and $\Psi_{\Gamma}$, $\Phi_{\Gamma}$ are the graph polynomials defined in eqs. (\ref{eq_def_Kirchhoff}, \ref{eq_def_2ndSym}). It can be found from the momentum space integrand by applying the Schwinger trick
	\al{
		\frac{1}{k_e^2}  = \int_0^{\infty} e^{-\alpha_ek_e^2} \mathrm d\alpha_e
	}
	to all propagators (i.e. terms corresponding to edges in the Feynman graph) $1/k_e^2$, integrating all momenta $k_e$ by simple Gaussian integration and collecting the remaining terms into the well-known polynomials via the matrix-tree theorem \cite{BognerWeinzierl_2010_FeynmanPolynomials}. Alternatively one could follow the more abstract geometric derivation in \cite[Sec. 6]{BlochEsnaultKreimer_2006_Motives}.

% The problem with quantum electrodynamics
\subsection{The problem with quantum electrodynamics}
For gauge theories the integrand becomes more complicated. Specifically, in quantum electrodynamics there are two different types of edges - fermions (\fermion) and photons (\photon) -  both of which are more complicated than scalar edges. Moreover, vertices give their own contributions to the integrand:
	\al{
		\photon  &= \frac{ g^{\nu_u\nu_v} + \varepsilon\frac{k^{\nu_u}_ek^{\nu_v}_e}{k_e^2} }{ k_e^2 }, \quad \text{ for } (u,v) = \partial(e) \\[3mm]
		\fermion &= \frac{ \gamma_{\mu_e} k_e^{\mu_e} }{ k_e^2 }\\[3mm]
		\vertexQED   &= \gamma_{\nu_v}
	}
	Thus one has to use a more complicated version of the Schwinger trick (cf. \cite{KreimerSarsSuijlekom_2013_QuantGauge, Sars_2015_PhD}):
	\al{
		\frac{ \gamma_{\mu_e} k_e^{\mu_e} }{ k_e^2 } = \gamma_{\mu_e} k_e^{\mu_e} \int_0^{\infty} e^{-\alpha_ek_e^2} \mathrm d\alpha_e
		\ \to\ -\gamma_{\mu_e} \int_0^{\infty} \frac{1}{2\alpha_e} \frac{\partial}{\partial \xi_{e,\mu_e}} e^{-\alpha_e(k_e+\xi_e)^2} \mathrm d\alpha_e
	}
	Note that we use ``$\to$'' instead of ``$=$'' in the second step. The introduction of the auxiliary momenta $\xi_e$ formally changes the expression and to get an equality one would need to evaluate them as explained in remark \ref{rem_evalSymanzik}. For theoretical computations it is useful to keep them and only replace them when the physical external momenta are explicitly needed (e.g. for renormalisation). Similarly, for a photon propagator one uses
	\al{\nonumber
		\frac{ g^{\nu_u\nu_v} + \varepsilon\frac{k^{\nu_u}_ek^{\nu_v}_e}{k_e^2} }{ k_e^2 }
		&= \int_0^{\infty}  \left( g^{\nu_u\nu_v} + \varepsilon\alpha_ek^{\nu_u}_ek^{\nu_v}_e\right) e^{-\alpha_ek_e^2}\mathrm d\alpha_e\\[3mm]
		&\to\ \int_0^{\infty}  \left(  \frac{2+\varepsilon}{2}g^{\mu_u\mu_v} + \frac{\varepsilon}{4\alpha_e}\frac{\partial^2}{\partial\xi_{e,\mu_u}\partial\xi_{e,\mu_v}} \right) e^{-\alpha_e(k_e+\xi_e)^2}\mathrm d\alpha_e.
	}
	After these replacements all $k_e$ are isolated in the exponential, so the same steps as in the scalar case can be applied to introduce the two graph polynomials. All in all, the integrand for a QED Feynman graph can thus be written as
	\al{
		I_{\Gamma} =  \gamma_{\Gamma}\mathrm D_{\Gamma}\frac{e^{-\frac{\Phi_{\Gamma}}{\Psi_{\Gamma}}}}{\Psi_{\Gamma}^2},\label{eq_QEDint}
	}
	where we used the abbreviation
	\al{
			\mathrm D_{\Gamma} & = \left( \prod_{e \in E^{(f)}_{\Gamma}} \left( -\frac{1}{2\alpha_e}\frac{\partial}{\partial\xi_{e,\mu_e}}  \right)\right) \prod_{\substack{e \in E^{(p)}_{\Gamma}\\ \partial(e)=(u,v)}} \left(  \frac{2+\varepsilon}{2}g^{\mu_u\mu_v} + \frac{\varepsilon}{4\alpha_e}\frac{\partial^2}{\partial\xi_{e,\mu_u}\partial\xi_{e,\mu_v}}   \right).
	}
	Meanwhile, the $\gamma^{\mu}$ appearing in the vertex and fermion terms, are the Dirac gamma matrices, satisfying
	\al{
		\gamma^{\mu}\gamma^{\nu} + \gamma^{\nu}\gamma^{\mu} = 2g^{\mu\nu}\mathds{1}_{4\times4}
	}
	and do not interest us in this article. We collect them into the abbreviation $\gamma_{\Gamma}$ that consists of a certain products of Dirac matrices for each ``open'' fermion line and traces
	\al{
		\prod_{C \in \mathcal C_{\Gamma}^{(f)}} \tr\left( \prod_{\substack{e \in C\\\partial_+(e)=w}} \gamma^{ \mu_w}\gamma^{\mu_e} \right)
	}
	where $\mathcal C_{\Gamma}^{(f)} \defeq \{ C \in \mathcal C_{\Gamma}^{[1]} \ | \ C \subset E_{\Gamma}^{(f)} \}$ contains all fermion cycles, $\partial_+(e) \in V_{\Gamma}$ is the vertex that $e$ is directed towards and the product over the edges in each cycle respects their ordering and goes opposite the direction defined by the fermion flow. For more detail on how to treat these objects see \cite{Kahane_1967_GammaAlgorithm, Chisholm_1972_GammaAlgGen}, or \cite{Golz_2017_Traces}, where we reinterprete, simplify and generalise Kahane's and Chisholm's algorithms.

% MAIN RESULT
\section{Main Result}\label{Sec_MainResult}

% Statement and proof
\subsection{Statement and proof}
	We can now express the result of the above mentioned differential operator $\mathrm D_{\Gamma}$ applied to $\exp(\Phi_{\Gamma}/\Psi_{\Gamma})$ in terms of our cycle polynomials and only need to introduce some more notational conventions.\\

	There is another polynomial that occurs, which is again a combination of certain Kirchhoff and cycle polynomials, namely
	\al{
		\Chi_{\Gamma}^{e,\mu}(\alpha,\xi) \defeq \xi_e^{\mu} \Psi_{\Gamma\dslash e} - \sum_{\substack{e' \in E_{\Gamma}\\ e'\neq e\ \; }} \xi_{e'}^{\mu}\alpha_{e'} \chi_{\Gamma}^{(e,e')}(\alpha).
	}
	Furthermore, the result involves a sum over all \textit{pairings} of a set. Let $S'$ be a set with $|S'|$ even and $\mathcal P_2(S')$ the set of partitions of $S'$ into parts of size 2. Then for any finite set $S$ the set of all pairings is given by 
	\al{
		\mathcal P(S) = \bigcup_{\substack{S'\subseteq S\\ |S'| \text{ even}}}\mathcal P_2(S').
	}
	For any pairing $P\in \mathcal P_2(S')$ we define $S_P = S\setminus S'$. As an example, consider $S=\{1,2,3,4,5\}$. Its even subsets are
	\al[*]{
		&\hspace{45mm} \{\},\\
		\{1,2\}, &\{1,3\}, \{1,4\}, \{1,5\}, \{2,3\}, \{2,4\}, \{2,5\}, \{3,4\}, \{3,5\}, \{4,5\},\\ 
		&\{1,2,3,4\}, \{1,2,3,5\}, \{1,2,4,5\}, \{1,3,4,5\}, \{2,3,4,5\}
	}
	and, for example, 
	\al{
		\mathcal P_2(\{1,2,3,4\}) &= \{ \{ (1,2), (3,4) \}, \{ (1,3), (2,4) \}, \{ (1,4), (2,3) \} \},\\[3mm]
		S_{\{(1,2),(3,4)\}} &= S\setminus \{1,2,3,4\} = \{5\}.
	}
	Some notational complication arises due to the mixing of vertices and edges in the formulae. We circumvent the problem by defining a map that is the identity for fermion edges and assigns to a vertex the unique photon edge incident to it, i.e.
	 $\bar e:  E_{\Gamma}^{(f)}\cup V_{\Gamma} \to E_{\Gamma}$ with
	\al[*]{
		\bar e (e) = e \quad \text{if } e\in E_{\Gamma}^{(f)} \hspace{30mm} \bar e (v) = e' \in \partial^{-1}(v \times V_{\Gamma}) \cap E_{\Gamma}^{(p)} \quad \text{if } v\in V_{\Gamma}		
	}
	Finally, in the following we only need those vertices that have an internal photon edge incident to it, not those whose photon edge is an external half edge. Hence we define the vertex subset
	\al{
		V_{\Gamma}^{(p)} \defeq \{ v \in V_{\Gamma} \ | \ \bar e (v) \neq \emptyset  \} = V_{\Gamma} \setminus \ker \bar e.
	}
	
	\begin{theo}\label{theo_main_numPoly}
		Let $\Gamma$ be a Feynman graph in quantum electrodynamics and $\mathrm D_{\Gamma}$ the differential operator as introduced above. Then
		\al[*]{
			\mathrm D_{\Gamma}e^{-\frac{\Phi_{\Gamma}}{\Psi_{\Gamma}}} = N_{\Gamma}e^{-\frac{\Phi_{\Gamma}}{\Psi_{\Gamma}}}
		}
		where
		\al{
			N_{\Gamma} =  \left(\prod_{e\in E_{\Gamma}^{(p)}}\varepsilon\alpha_e\right) \sum_{P \in\mathcal P(E_{\Gamma}^{(f)} \cup V_{\Gamma}^{(p)})}\left( \prod_{(i,j)\in P} \frac{g^{\mu_{i}\mu_{j}}}{2} \frac{\chi_{\Gamma}^{(\bar e(i), \bar e(j))}}{\Psi_{\Gamma}}  \right) \left(\prod_{k\in (E_{\Gamma}^{(f)} \cup V_{\Gamma}^{(p)})_P} \frac{\Chi_{\Gamma}^{\bar e(k),\mu_{k}}}{\Psi_{\Gamma}} \right).\\\nonumber
		}
	\end{theo}
	Note that whenever $i$ and $j$ are the endpoints of the same photon edge one has $\bar e(i) = e = \bar e(j)$ such that one gets $\chi_{\Gamma}^{(e,e)} = \chi_{\Gamma}^{(e)} + 2\Psi_{\Gamma}/\varepsilon\alpha_e$ as defined in \refeq{eq_chi_sameedge}. Thus, in the case of Feynman gauge $\varepsilon \to 0$ the result simplifies to
	\al{
		N_{\Gamma} = \Bigg( \prod_{\substack{e \in E^{(p)}_{\Gamma}\\ \partial(e)=(u,v)}} g^{\mu_u\mu_v}\Bigg)  \sum_{P \in\mathcal P(E_{\Gamma}^{(f)})}\left( \prod_{(i,j)\in P} \frac{g^{\mu_{i}\mu_{j}}}{2} \frac{\chi_{\Gamma}^{(i,j)}}{\Psi_{\Gamma}}  \right) \left(\prod_{k\in (E_{\Gamma}^{(f)})_P} \frac{\Chi_{\Gamma}^{k,\mu_{k}}}{\Psi_{\Gamma}} \right).
	}	
	\begin{proof}[Proof of Theorem \ref{theo_main_numPoly}]
	Consider first the derivative of the second Symanzik polynomial. By inserting the definition from \refeq{eq_def_2ndSym} one finds
		\al{\nonumber
			\frac{1}{2}\frac{\partial}{\partial \xi_{e,\mu}} \Phi_{\Gamma}(\alpha, \xi) 
				& = \frac{1}{2}\frac{\partial}{\partial \xi_{e,\mu}} \sum_{B \in \mathcal B_{\Gamma}} \left(\sum_{e' \in B} \mathrm o_{B}(e') \xi_{e'}  \right)^2 \alpha_B \Psi_{\Gamma\setminus B}(\alpha)\\[3mm]
				& = \sum_{B \in \mathcal B_{\Gamma}}  \mathrm o_{B}(e) \left(\sum_{e' \in B} \mathrm o_{B}(e') \xi_{e'}^{\mu}  \right) \alpha_B\Psi_{\Gamma\setminus B}(\alpha).
		}
		We reorder the terms in the double sum and collect coefficients of $\xi_{e'}^{\mu}$ for all edges $e'$. This gives
		\al{\nonumber
			\frac{1}{2}\frac{\partial}{\partial \xi_{e,\mu}} \Phi_{\Gamma}(\alpha, \xi) 
			&= \sum_{e' \in E_{\Gamma}} \xi_{e'}^{\mu} \sum_{B \in \mathcal B_{\Gamma}}  \mathrm o_{B}(e) \mathrm o_{B}(e') \alpha_B\Psi_{G\setminus B}(\alpha)\\[3mm]\nonumber
			&= \xi_{e}^{\mu} \sum_{B \in \mathcal B_{\Gamma}}  \mathrm o_{B}(e)^2 \alpha_B\Psi_{\Gamma\setminus B}(\alpha) +\sum_{\substack{e' \in E_{\Gamma}\\e'\neq e}} \xi_{e'}^{\mu} \sum_{B \in \mathcal B_{\Gamma}}  \mathrm o_{B}(e,e') \alpha_B\Psi_{\Gamma\setminus B}(\alpha)\\[3mm]
			&= \xi_{e}^{\mu} \beta_{\Gamma}^{(e)}(\alpha)
			+\sum_{\substack{e' \in E_{\Gamma}\\e'\neq e}} \xi_{e'}^{\mu} \beta_{\Gamma}^{(e,e')}(\alpha)
		}
	Thus, by lemma \ref{lem_bondpol_decomp} for the first term and lemma \ref{lem_bondcycle} for the sum one has
	\al{
			\frac{1}{2\alpha_e}\frac{\partial}{\partial \xi_{e,\mu}} \Phi_{\Gamma}(\alpha, \xi) 
			= \xi_e^{\mu} \Psi_{\Gamma\dslash e}(\alpha) - \sum_{\substack{e' \in E_{\Gamma}\\e'\neq e}} \xi_{e'}^{\mu}\alpha_{e'} \chi_{\Gamma}^{(e,e')}(\alpha) = \Chi_{\Gamma}^{e,\mu}(\alpha,\xi).
	}
	With the help of this equation we can quickly compute the derivatives of $e^{-\frac{\Phi_{\Gamma}}{\Psi_{\Gamma}}}$ that appear in $\mathrm D_{\Gamma}$. For a single photon $e\in E_{\Gamma}^{(p)}$ between vertices $u,v\in V_{\Gamma}$ one finds
	\al{\nonumber
		\left(  \frac{2+\varepsilon}{2}g^{\mu_u\mu_v}\right. & \left. + \frac{\varepsilon}{4\alpha_e}\frac{\partial^2}{\partial\xi_{e,\mu_u}\partial\xi_{e,\mu_v}} \right)e^{-\frac{\Phi_{\Gamma}}{\Psi_{\Gamma}}}\\[5mm]\nonumber
		& =  \left( \frac{2+\varepsilon}{2}g^{\mu_u\mu_v} + \varepsilon\alpha_e\left( -\frac{g^{\mu_u\mu_v}\Psi_{\Gamma\dslash e}}{2\alpha_e\Psi_{\Gamma}} + \frac{\Chi_{\Gamma}^{e,\mu_u}\Chi_{\Gamma}^{e,\mu_v}}{\Psi_{\Gamma}^2} \right)\right) e^{-\frac{\Phi_{\Gamma}}{\Psi_{\Gamma}}}
		}
		\al{\nonumber
		\qquad & =  \left(g^{\mu_u\mu_v} \frac{2\Psi_{\Gamma}+\varepsilon(\Psi_{\Gamma}-\Psi_{\Gamma\dslash e})}{2\Psi_{\Gamma}} + \varepsilon\alpha_e\frac{\Chi_{\Gamma}^{e,\mu_u}\Chi_G^{e,\mu_v}}{\Psi_{\Gamma}^2}\right) e^{-\frac{\Phi_{\Gamma}}{\Psi_{\Gamma}}}\\[5mm]\nonumber
		& =  \varepsilon\alpha_e\left( g^{\mu_u\mu_v}\frac{\frac{2\Psi_{\Gamma}}{\varepsilon\alpha_e}+\Psi_{\Gamma\setminus e}}{2\Psi_{\Gamma}} + \frac{\Chi_{\Gamma}^{e,\mu_u}\Chi_{\Gamma}^{e,\mu_v}}{\Psi_{\Gamma}^2}\right) e^{-\frac{\Phi_{\Gamma}}{\Psi_{\Gamma}}}\\[5mm]
		& =  \varepsilon\alpha_e\left( \frac{g^{\mu_u\mu_v}\chi_{\Gamma}^{(e,e)}}{2\Psi_{\Gamma}} + \frac{\Chi_{\Gamma}^{e,\mu_u}\Chi_{\Gamma}^{e,\mu_v}}{\Psi_{\Gamma}^2}\right)e^{-\frac{\Phi_{\Gamma}}{\Psi_{\Gamma}}}.
	}
	The combined derivative for fermion edges $e_1, e_2\in E_{\Gamma}^{(f)}$ is
	\al{\nonumber
		\frac{1}{4\alpha_{e_1}\alpha_{e_2}}\frac{\partial^2}{\partial \xi_{e_1,\mu_1}\partial \xi_{e_2,\mu_2}} e^{-\frac{\Phi_{\Gamma}}{\Psi_{\Gamma}}} 
		& = -\frac{1}{2\alpha_{e_1}}\frac{\partial}{\partial \xi_{e_1,\mu_1}}\left( \frac{\Chi_{\Gamma}^{e_2,\mu_2}}{\Psi_{\Gamma}}  e^{-\frac{\Phi_{\Gamma}}{\Psi_{\Gamma}}}\right)\\[3mm]
		& = \left( \frac{g^{\mu_1\mu_2}\chi_{\Gamma}^{(e_1,e_2)}}{2\Psi_{\Gamma}} + \frac{\Chi_{\Gamma}^{e_1,\mu_1}\Chi_{\Gamma}^{e_2,\mu_2}}{\Psi_{\Gamma}^2} \right)e^{-\frac{\Phi_{\Gamma}}{\Psi_{\Gamma}}}.
	}
	By the Leibniz rule the desired result is then just a sum over certain combinations of these basic results, exactly given by the pairings introduced above.\\
	\end{proof}

% Examples
\subsection{Examples}
	Consider again the graphs $\Gamma_1$ and $\Gamma_2$ from \reffig{Example_02_FeynmanGraphs} for which we computed the Kirchhoff and second Symanzik polynomials in example \ref{exam_classGraphPol}. We will use Feynman gauge for $\Gamma_1$ and general gauge for the smaller $\Gamma_2$. For the sake of simpler notation we will in this section use $\nu_i$ as space-time index corresponding to a vertex $v_i$ and $\mu_i$ for those corresponding to edges $e_i$.
	% Example 1
	\begin{exam}\label{ex_main_1}
		$\Gamma_1$ has three simple cycles		
		\al[*]{
			C_1 = \{ e_1, e_2, e_5 \} \quad C_2 = \{ e_1, e_3, e_4 \} \quad C_3 = \{ e_2, e_3, e_4, e_5 \} 
		}
		which give the polynomials
		\al[*]{
			\Psi_{\Gamma_1 \dslash C_1} = \alpha_3 + \alpha_4, \qquad
			\Psi_{\Gamma_1 \dslash C_2} = \alpha_2 + \alpha_5, \qquad
			\Psi_{\Gamma_1 \dslash C_3} = \alpha_1.\quad
		}
		In Feynman gauge we only need pairings of fermion edges $\{ e_2, e_3, e_4, e_5 \}$. There are $\binom{4}{2} = 6$ with one pair, $3!! = 3$ with two pairs and the empty pairing. The relevant cycle polynomials are
		\al[*]{
			\chi_{\Gamma_1}^{(e_2,e_3)} & = \Psi_{\Gamma_1 \dslash C_3} = \alpha_1\qquad
			\chi_{\Gamma_1}^{(e_2,e_4)} = \Psi_{\Gamma_1 \dslash C_3} = \alpha_1\\
			\chi_{\Gamma_1}^{(e_3,e_5)} & = \Psi_{\Gamma_1 \dslash C_3} = \alpha_1 \qquad
			\chi_{\Gamma_1}^{(e_4,e_5)} = \Psi_{\Gamma_1 \dslash C_3} = \alpha_1\\
			\chi_{\Gamma_1}^{(e_2,e_5)} & = \Psi_{\Gamma_1 \dslash C_1} + \Psi_{\Gamma_1 \dslash C_3} = \alpha_1 + \alpha_3 + \alpha_4\\
			\chi_{\Gamma_1}^{(e_3,e_4)} & = \Psi_{\Gamma_1 \dslash C_2} + \Psi_{\Gamma_1 \dslash C_3} = \alpha_1 + \alpha_2 + \alpha_5
		}
		and one finds $N_{\Gamma_1} = N_{\Gamma_1}^{(0)} + N_{\Gamma_1}^{(1)} + N_{\Gamma_1}^{(2)}$, where
		\al{
			N_{\Gamma_1}^{(0)} = \frac{g^{\nu_2\nu_4}}{4\Psi_{\Gamma_1}^2}
			\left( (g^{\mu_2\mu_3}g^{\mu_4\mu_5} + g^{\mu_2\mu_4}g^{\mu_3\mu_5})\alpha_1^2 +  g^{\mu_2\mu_5}g^{\mu_3\mu_4}(\alpha_1 + \alpha_3 + \alpha_4)(\alpha_1 + \alpha_2 + \alpha_5)\right),
		}
		$N_{\Gamma_1}^{(1)}$ consists of six terms of the form
		\al{
			\frac{g^{\nu_2\nu_4}}{2\Psi_{\Gamma_1}^3} g^{\mu_2\mu_3}\chi_{\Gamma_1}^{(e_2,e_3)}\Chi_{\Gamma_1}^{e_4,\mu_4}\Chi_{\Gamma_1}^{e_5,\mu_5},
		}
		where e.g.
		\al{\nonumber
			\Chi_{\Gamma_1}^{e_5,\mu_5}& = \xi_5^{\mu_5}\big((\alpha_1+\alpha_2)(\alpha_3+\alpha_4) + \alpha_1\alpha_2\big)\\
			&\quad - \xi_2^{\mu_5}\alpha_2(\alpha_1 + \alpha_3 + \alpha_4) - \xi_3^{\mu_5}\alpha_1\alpha_3 - \xi_4^{\mu_5}\alpha_1\alpha_4,
		}
		and 
		\al{
			N_{\Gamma_1}^{(2)} = \frac{g^{\nu_2\nu_4}}{\Psi_{\Gamma_1}^4} \Chi_{\Gamma_1}^{e_2,\mu_2}\Chi_{\Gamma_1}^{e_3,\mu_3}\Chi_{\Gamma_1}^{e_4,\mu_4}\Chi_{\Gamma_1}^{e_5,\mu_5}.
		}
		Evaluating $\xi$ to get physical momenta simplifies matters. Choose $\{e_2,e_3\}$ as the momentum path. Then
		\al[*]{
			\Chi_{\Gamma_1}^{e_2,\mu_2} &= q^{\mu_2}\big((\alpha_1+\alpha_5)(\alpha_3+\alpha_4) + \alpha_1\alpha_5 -\alpha_1\alpha_3\big)\\[1mm]
			\Chi_{\Gamma_1}^{e_3,\mu_3} &= q^{\mu_3}\big((\alpha_1+\alpha_4)(\alpha_2+\alpha_5) + \alpha_1\alpha_4 -\alpha_1\alpha_2\big)\\[1mm]
			\Chi_{\Gamma_1}^{e_4,\mu_4} &= -q^{\mu_4}\big(  \alpha_1\alpha_2 + \alpha_3(\alpha_1 + \alpha_2 + \alpha_5)\big)\\[1mm]
			\Chi_{\Gamma_1}^{e_5,\mu_5} &= -q^{\mu_5}\big(  \alpha_2(\alpha_1 + \alpha_3 + \alpha_4) + \alpha_1\alpha_3\big)
		}
	\end{exam}
	
	% Example 2
	\begin{exam}\label{ex_main_2}
		$\Gamma_2$ contains only a single simple cycle that encompasses the entire graph, except for the external half-edges. Hence $\Psi_{\Gamma_2 \dslash C} = 1$ and the three cycle polynomials $\chi_{\Gamma_2}^{(e_1,e_2)}$, $\chi_{\Gamma_2}^{(e_1,e_3)}$ and $\chi_{\Gamma_2}^{(e_2,e_3)}$ are also all just $1$. Beyond that we have
		\al[*]{
			\Chi_{\Gamma_2}^{e_1,\nu_i}& = \xi_1^{\nu_i}(\alpha_2+\alpha_3) - \xi_2^{\nu_i}\alpha_2- \xi_3^{\nu_i}\alpha_3
			\ \, \quad \to \quad  q_2^{\nu_i}(\alpha_2+\alpha_3) - q_1^{\nu_i}\alpha_3  \qquad i=2,3\\[1mm]
			\Chi_{\Gamma_2}^{e_2,\mu_2}& = \xi_2^{\mu_2}(\alpha_1+\alpha_3) - \xi_1^{\mu_2}\alpha_1- \xi_3^{\mu_2}\alpha_3
			\quad \to \quad  -q_2^{\mu_2}\alpha_1- q_1^{\mu_2}\alpha_3  \\[1mm]
			\Chi_{\Gamma_2}^{e_3,\mu_3}& = \xi_3^{\mu_3}(\alpha_1+\alpha_2) - \xi_1^{\mu_3}\alpha_1- \xi_2^{\mu_3}\alpha_2
			\quad \to \quad  q_1^{\mu_3}(\alpha_1+\alpha_2) - q_2^{\mu_3}\alpha_1 \\[3mm]
			&\hspace{4cm}	\chi_{\Gamma_2}^{(e_1,e_1)} = 1 + \frac{2\Psi_{\Gamma_2}}{\varepsilon\alpha_1}
		}
		In this example we again have to consider pairings of a set of four objects, but this time we include vertices and have to remember $\bar e(v_2) = e_1 = \bar e(v_3)$. We again have $N_{\Gamma_2} = N_{\Gamma_2}^{(0)} + N_{\Gamma_2}^{(1)} + N_{\Gamma_2}^{(2)}$, where the summands are constructed from the polynomials above in the same way as in the last example. In order to illustrate the difference between Feynman gauge and general gauge consider the first term:
		\al{\nonumber
			N_{\Gamma_2}^{(0)} &= \frac{\varepsilon\alpha_1}{4\Psi_{\Gamma_2}^2}\left( g^{\mu_2\mu_3}g^{\nu_2\nu_3}\left(1 + \frac{2\Psi_{\Gamma_2}}{\varepsilon\alpha_1}\right) + g^{\mu_2\nu_2}g^{\mu_3\nu_3} + g^{\mu_2\nu_3}g^{\mu_3\nu_2}\right)\\[3mm]
			& = \frac{g^{\mu_2\mu_3}g^{\nu_2\nu_3}}{2\Psi_{\Gamma_2}} 
			+ \varepsilon\alpha_1 \frac{g^{\mu_2\mu_3}g^{\nu_2\nu_3} + g^{\mu_2\nu_2}g^{\mu_3\nu_3} + g^{\mu_2\nu_3}g^{\mu_3\nu_2}}{4\Psi_{\Gamma_2}^2}
		}
		
	\end{exam}

% CONCLUSION
\section*{Conclusion}
	With our main theorem we have found an explicit combinatorial expression for the rational function that differentiates the parametric integrand of quantum electrodynamics from the scalar case. This expression contains a new graph polynomial based on cycle subgraphs whose properties we studied. The occurrence of this polynomial and especially the interesting properties discussed in the lemmata used to prove the main theorem suggest a deeper importance of cycle subgraphs in gauge theory, which is in accordance with \cite{KreimerSarsSuijlekom_2013_QuantGauge}, where a combination of graph and cycle homology was used to derive the parametric integrand for general gauge theories. Since the properties of the cycle polynomial and even its appearance in the derivatives of $\Phi_{\Gamma}(\alpha, \xi)$ is independent of the specific case of QED it should be possible to generalise our results to general gauge theories by replacing the special case of $\mathrm D_{\Gamma}$ with the general Corolla differential \cite{KreimerYeats_2012_Corolla, KreimerSarsSuijlekom_2013_QuantGauge, Prinz_2016_Corolla}. The cycle polynomial and generation of the integrand has been implemented and checked via computer algebra for all photon, fermion and vertex functions up to two loops and some three-loop photon functions. While the computational advantage of our result compared to naive differentiation is at this point minor at best, it is an important step towards a simplification of gauge theory integrands. The combinatoric understanding of the integrand we gained in this article can be used in conjunction with the results of \cite{Golz_2017_Traces} to yield an integrand in which all Dirac matrices and metric tensors have been contracted. While the precise structure of the integrand after contraction still stands to be investigated in future work we believe that it, too, should have a relatively simple combinatorial interpretation along the lines of the work presented in this article, which would be an inestimable improvement from both practical and theoretical viewpoints.

% References
	\bibliographystyle{plainurl}
	\bibliography{userBib}

\end{document}